\documentclass[final,3p,times]{elsarticle}

\usepackage{amsmath,amsthm,amssymb,amsfonts,amscd}
\usepackage{graphicx}
\usepackage{xspace}
\usepackage{color,soul}

\def\N{\mathbb N}
\def\C{\mathbb C}

\renewcommand{\i}{{\mathrm i}}
\newcommand\e{\mathrm e}

\newcommand\diag{\mathrm{diag}}
\renewcommand\S{\mathcal S}

\newcommand\tg{\mathrm{tg}\,}
\newcommand\cotg{\mathrm{cotg}\,}
\newcommand\PS{\textsf{PS}\xspace}
\newcommand\MPS{\textsf{MPS}\xspace}
\newtheorem{lemma}{Lemma}[section]
\newtheorem{proposition}[lemma]{Proposition}

\newtheorem{corollary}[lemma]{Corollary}
\newtheorem{observation}[lemma]{Observation}
\newtheorem{definition}[lemma]{Definition}

\newtheorem{remark}[lemma]{Remark}
\newtheorem{notation}[lemma]{Notation}

\begin{document}

\begin{frontmatter}

\title{Quantum graph vertices with permutation-symmetric scattering probabilities}
\author
{Ond\v{r}ej Turek}
\ead{ondrej.turek@kochi-tech.ac.jp}
\author
{Taksu Cheon}
\ead{taksu.cheon@kochi-tech.ac.jp}
\address
{Laboratory of Physics, Kochi University of Technology,
Tosa Yamada, Kochi 782-8502, Japan}

\date{\Today}

\begin{abstract}
Boundary conditions in quantum graph vertices are generally given in terms of a unitary matrix $U$. Observing that if $U$ has at most two eigenvalues, then the scattering matrix $\S(k)$ of the vertex is a linear combination of the identity matrix and a fixed Hermitian unitary matrix, we construct vertex couplings with this property: For all momenta $k$, the transmission probability from the $j$-th edge to $\ell$-th edge is independent of $(j,\ell)$, and all the reflection probabilities are equal. We classify these couplings according to their scattering properties, which leads to the concept of generalized $\delta$ and $\delta'$ couplings.
\end{abstract}

\begin{keyword}
scattering matrix \sep singular vertex
\sep equal transmission \sep quantum wire
\PACS 03.65.-w \sep 03.65.Nk \sep 73.21.Hb%
\end{keyword}

\end{frontmatter}


\section{Introduction}

The study of quantum graph, the quantum mechanics of particle motion on interconnected one-dimensional lines, is becoming increasingly relevant to laboratory experiments on quantum wires and quantum dots, as experimentalists are gaining more control over the electronic states in these artificial objects in atomic scale.
The designing of quantum graphs tailor-made to have desired properties is now high on our agenda.  It amounts to exploring the proper characterization of the scattering matrices of quantum graphs, and identifying the structure of graph vertices that give rise to those scattering matrices. 

As a starting point for the designing of quantum graphs, we have considered, in our previous work \cite{TC11, CET11}, the momentum-independent equally-transmitting quantum graph vertex, for which incoming wave from any graph edge results in equal reflections and also equal transmissions to all other edges.  The momentum independence is guaranteed by limiting the graph vertices to scale invariant ``Fulop-Tsutsui'' type, which simplifies the situation considerably, since the scattering matrices of such vertices are Hermitian unitary \cite{CT10}.  The simplification has yielded a mathematical reward in that we have been able to connect the designing {\it equally-transmitting quantum vertices} to the concept of modularly permutation symmetric unitary matrix, which amounts to an extension of Hadamard and conference matrices.
This simplification, however, has its own drawbacks:  The considered class does not include such natural equally-transmitting vertices as $\delta$ coupling and also  $\delta'$ coupling.  For their inclusion, we naturally have to go beyond the scale invariant subfamily of vertices, and consider more general momentum-dependent scattering matrices.

In this article, we examine a new class of more general incoming-momentum dependent equally-transmitting quantum vertices that include $\delta$ and $\delta'$ couplings.  To that end, we take a slightly different path from our previous work, in which we have relied on ``ST-form'' of connection matrix \cite{CET10}, and utilize Harmer-Kostrykin-Schrader unitary-matrix characterization instead.  
It turns out that we obtain an advantage by further extending the equally-transmitting graph with quantum vertex described by unitary matrices having at most two eigenvalues.  We show that such quantum vertex is a natural generalization of scale invariant vertex, in the sense that its scattering matrix is characterized by shifted and scaled Hermitian unitary matrix. It is shown that, with further delimitations, we obtain several subclasses of incoming-momentum dependent equally-transmitting quantum vertices that correspond to the generalizations of $\delta$, $\delta'$, and ``mixed'' couplings.

The article is organized as follows:  In the next section, we lay out preliminary materials to set the stage.  The quantum vertices described by ``quadratically closed''  unitary matrices are introduced in the third section, and their scattering matrices are analyzed in the fourth section.  In the fifth section, quantum vertices with generalized permutation symmetry are derived as an important special class of vertices studied in previous sections, and the concepts of generalized $\delta$ and $\delta'$ couplings are introduced.  The sixth section concludes the article.


\section{Preliminaries}

\subsection*{Boundary conditions and scattering matrix}

Vertex couplings in quantum 
star graphs are mathematically given by boundary conditions which couple the vectors $\Psi$ and $\Psi'$ of the boundary values. The most general form of these conditions is
\begin{equation}\label{U-form}
(U-I)\Psi+\i(U+I)\Psi'=0\,,
\end{equation}
where $U$ is a unitary matrix (for a derivation see \cite{Ha00} and \cite{KS00}).

A vertex coupling can be uniquely determined also by its \emph{scattering matrix} $\S(k)$. $\S(k)$ is a unitary energy-dependent matrix (more precisely speaking, a matrix function),
\begin{equation}
\label{ScatMat}
\S(k)=\left(\begin{array}{cccc}
 {\cal R}_{1}(k)   & {\cal T}_{12}(k) & \cdots & {\cal T}_{1n}(k) \\
 {\cal T}_{21}(k) & {\cal R}_{2}(k)   & \cdots & {\cal T}_{2n}(k) \\
 \vdots & & & \vdots \\
 {\cal T}_{n1}(k) & {\cal T}_{n2}(k)& \cdots & {\cal R}_{n}(k) \end{array}\right)\,,
\end{equation}
where $n$ is the degree of the vertex, $\mathcal{R}_j(k)$ is the reflection amplitude for the $j$-th line and $\mathcal{T}_{j\ell}(k)$ is the transmission amplitude from the $\ell$-th to the $j$-th line. Here and elsewhere, the parameter $k$ denotes the square root of the energy of the particle coming into the vertex and at the same time the particle momentum, as we assume for simplicity $\hbar=2m=1$.

The matrix $U$ and the scattering matrix $\S(k)$ are related in the following way: $U=\S(1)$ and
\begin{equation}\label{S-matrix}
\S(k)=\bigl[(k-1)U+(k+1)I\bigr]^{-1}\cdot\bigl[(k+1)U+(k-1)I\bigr]\,.
\end{equation}

\begin{remark}\label{decoupled}
Matrix $U$ is diagonal if and only if $\S(k)$ is diagonal for all $k$. This singular case physically represents a vertex whose edges are all 
isolated (a \emph{decoupled vertex}).
\end{remark}

\subsection*{\PS and \MPS unitary matrices}

\begin{notation}\label{IJ}
Everywhere in the paper, $I$ and $J$ denote an identity matrix and a square matrix all of whose entries are $1$, respectively. In cases when we need to explicitely express the order $\ell$ of these matrices, we use the symbols $I^{(\ell)}$ and $J^{(\ell)}$.
\end{notation}

\begin{definition}\label{MPS}
\textit{(i)}\quad A square matrix $M\in\C^{n,n}$ is called \emph{permutation-symmetric} (\PS) if there are $r,t\in\C$ such that the entries of $M$ satisfy
$$
M_{jj}=r \quad \text{and} \quad M_{j\ell}=t \quad \text{for all $j,\ell=1,\ldots,n$, $j\neq \ell$}\,.
$$
\textit{(ii)}\quad We call a square matrix $M\in\C^{n,n}$ \emph{modularly permutation-symmetric} (\MPS) if there are $r,t\geq0$ such that the entries of $M$ satisfy
$$
|M_{jj}|=r \quad \text{and} \quad |M_{j\ell}|=t \quad \text{for all $j,\ell=1,\ldots,n$, $j\neq \ell$}\,.
$$
\end{definition}

\begin{remark}
If a matrix $M$ is \PS (or \MPS) and $P$ is a permutation matrix of the same size, then the matrix $PMP^{-1}$ is \PS (or \MPS, respectively) as well.
\end{remark}

\PS and especially \MPS matrices that are at the same time \emph{unitary} are important in various applications. It is easy to characterize all unitary \PS matrices (cf.~\cite{ET06}):
\begin{proposition}\label{aIbJ}
A unitary $n\times n$ matrix $U$ is permutation-symmetric if and only if $U=aI^{(n)}+bJ^{(n)}$ for $a,b\in\C$ satisfying $|a|=1$ and $|a+nb|=1$.
\end{proposition}
On the other hand, full characterization of unitary \MPS matrices is a difficult problem and only partial results are known.

For any non-diagonal \PS or \MPS matrix we set $d:=\frac{|M_{jj}|}{|M_{j\ell}|}$ ($j\neq\ell$). In our further considerations, an essential role will be played by (non-diagonal) \emph{Hermitian} unitary \MPS matrices, i.e., by unitary $n\times n$ matrices of the type
$$
M=\frac{1}{\sqrt{d^2+n-1}}\left(\begin{array}{ccccc}
\pm d & \e^{\i\alpha_{12}} & \e^{\i\alpha_{13}} & \cdots & \e^{\i\alpha_{1n}} \\
\e^{-\i\alpha_{12}} & \pm d & \e^{\i\alpha_{23}} & \cdots & \e^{\i\alpha_{2n}} \\
\e^{-\i\alpha_{13}} & \e^{-\i\alpha_{23}} & \pm d & \cdots & \e^{\i\alpha_{3n}} \\
\vdots & \vdots & & \ddots & \vdots \\
\e^{-\i\alpha_{1n}} & \e^{-\i\alpha_{2n}} & \e^{-\i\alpha_{3n}} & \cdots & \pm d
\end{array}\right)\,,
$$
where $\alpha_{j\ell}$ are real numbers representing phases.

\subsection*{Permutation-symmetric couplings}

As we will see, \PS and \MPS unitary matrices are essential objects in the study of vertex couplings with equal scattering probabilities.

\begin{definition}\label{PScoupl}
We call a vertex coupling \emph{permutation-symmetric} (\PS) if the corresponding boundary conditions~\eqref{U-form} are equivalent to
$$
(U-I)P\Psi+\i(U+I)P\Psi'=0
$$
for any permutation matrix $P\in\mathcal{S}_n$.
\end{definition}
In other words, the coupling is permutation-symmetric if the vertex is invariant under renumbering of its edges.
Equivalent definitions can be easily obtained using Equation~\eqref{S-matrix} and Proposition~\ref{aIbJ}:

\begin{observation}\label{PS U}
Let a vertex coupling be given by boundary conditions~\eqref{U-form}. The following statements are equivalent.
\begin{itemize}
\item The vertex coupling is permutation-symmetric.
\item The matrix $U$ is \PS.
\item $U=aI^{(n)}+bJ^{(n)}$ for $a,b\in\C$ satisfying $|a|=1$ and $|a+nb|=1$.
\item The scattering matrix $\S(k)$ is \PS for all $k>0$.
\end{itemize}
\end{observation}

All the permutation-symmetric couplings form a $2$-parameter family that has been investigated in detail in \cite{ET06}. The family includes the physically most important vertex couplings, in particular the $\delta$ and $\delta'$ couplings, as it can be seen from Table~\ref{DeltaCouplings} (see also~\cite{Ex96a}).

\begin{table}[h]
\caption{Special permutation-symmetric vertex couplings ($U=aI^{(n)}+bJ^{(n)}$)}\label{DeltaCouplings}
\begin{tabular}{c|cc}
Coupling type & a & b \\[2pt]
\hline
\\[-7pt]
free coupling & $-1$ & $\frac{2}{n}$ \\[5pt]
$\delta$-coupling with parameter $\gamma$ & $-1$ & $\frac{2}{n+\i\gamma}$ \\[5pt]
$\delta'_s$-coupling with parameter $\gamma'$ & $1$ & $-\frac{2}{n-\i\gamma'}$ \\[5pt]
$\delta'$-coupling with parameter $\gamma'$ & $-\frac{n+\i\gamma'}{n-\i\gamma'}$ & $\frac{2}{n-\i\gamma'}$ \\[5pt]
$\delta_p$-coupling with parameter $\gamma$ & $\frac{n-\i\gamma}{n+\i\gamma}$ & $-\frac{2}{n+\i\gamma}$
\end{tabular}
\end{table}

\begin{observation}\label{PS2ev}
If $U$ is \PS, then it has at most two eigenvalues.
\end{observation}

\begin{proof}
Since $\sigma(J^{(n)})=\{0,n\}$, we have $\sigma(U)=\{a,a+bn\}$, hence $\#\sigma(U)\leq2$.
\end{proof}

\subsection*{Scale invariant couplings}

The stattering matrix generaly depends on the momentum $k$, but there exists a family of vertex couplings for which $\S(k)$ is $k$-independent. These couplings are called \emph{scale invariant}. With regard to Equation~\eqref{S-matrix}, scale invariant vertex couplings can be characterized in several ways:

\begin{observation}\label{MPSequiv}
The following statements are equivalent.
\begin{itemize}
\item The vertex coupling is scale invariant.
\item The vertex coupling be given by b.c.~\eqref{U-form} with a Hermitian matrix $U$.
\item The scattering matrix $\S(k)$ is Hermitian for all $k>0$.
\end{itemize}
\end{observation}

\subsection*{Equally-transmitting couplings}

Let us finish Preliminaries by introducing one more family of vertex couplings which is, in some respect, a generalization of permutation-symmetric couplings.

\begin{definition}\label{MPScoupl}
We say that a vertex coupling is \emph{equally-transmitting}\footnote{The term ``equally-transmitting'' (vertex or coupling) is not to be confused with the term ``equi-transmitting'' (matrix) introduced by Harrison, Smilansky and Winn in~\cite{HSW07}. Equi-transmitting matrix from~\cite{HSW07} is, in our notation, a unitary MPS matrix with $d=0$.} if its scattering matrix $\S(k)$ is \MPS for all $k>0$.
\end{definition}

An important parameter of any equally-transmitting vertex coupling (except for the decoupled vertex, see Rem.~\ref{decoupled}) is the ratio
\begin{equation}\label{rho}
\rho(k)=\frac{|\S(k)_{jj}|^2}{|\S(k)_{j\ell}|^2} \qquad (j\neq\ell)
\end{equation}
which has the following physical meaning:
$$
\rho(k)=\frac{\text{reflection probability at momentum $k$}}{\text{transmission probability at momentum $k$}}\,.
$$

Equally-transmitting couplings form a much richer family than \PS couplings do, and there does not exist so far a simple characterization of them. Apparently, Observation~\ref{PS U} \emph{cannot} be straightforwardly extended to them. The examination of equally-transmitting couplings has been initiated probably by the paper~\cite{HSW07} whose authors studied vertex couplings with these two properties: (i) $\S(1)$ is \MPS, (ii) $\diag(\S(1))=0$, in other words, couplings which are uniformly transmitting and non-reflecting at $k=1$. Later, in paper~\cite{TC11}, we 
have examined scale invariant equally-transmitting vertex couplings. Section~\ref{GPS couplings} of this paper will be devoted to an extension of these results.


\section{Vertex couplings for $U$ having at most $2$ eigenvalues}

With regard to Observations~\ref{PS2ev} and~\ref{MPSequiv}, vertex couplings given by $U$ having at most two eigenvalues seem to form a natural common generalization of permutation-symmetric couplings and scale invariant couplings.

For any unitary matrix $U$ having exactly two eigenvalues, $\sigma(U)=\{\e^{\i\alpha},\e^{\i\beta}\}$, there exist orthogonal projectors $P$ ($P\neq0$) and $Q=I-P$ such that
\begin{equation}\label{U(P)}
U=\e^{\i\alpha}P+\e^{\i\beta}Q=(\e^{\i\alpha}-\e^{\i\beta})P+\e^{\i\beta}I\,.
\end{equation}
In case that $U=\{\e^{\i\alpha}\}$, Equation~\eqref{U(P)} can be used as well after setting $P=I$ and taking an arbitrary $\beta$.

Since $P$ is Hermitian and $\sigma(P)\subset\{0,1\}$, the matrix $M$ defined as $M:=2P-I$ is Hermitian and satisfies $\sigma(M)\subset\{-1,1\}$, thus $M$ is Hermitian unitary.

\begin{proposition}\label{U M}
If $U$ is a unitary matrix having at most two eigenvalues, then there exist $\alpha,\beta\in(-\pi,\pi]$ and a Hermitian unitary matrix $M$ such that
\begin{equation}\label{U(M)}
U=\e^{\i\frac{\alpha+\beta}{2}}\cdot\left(\cos\frac{\alpha-\beta}{2}\cdot I+\i\sin\frac{\alpha-\beta}{2}\cdot M\right)\,.
\end{equation}
\end{proposition}

\begin{proof}
Let $U$ be given by~\eqref{U(P)}, and $M=2P-I$. Then $P=\frac{1}{2}(M+I)$ and a straightforward calculation gives
\begin{equation*}
U=\e^{\i\frac{\alpha+\beta}{2}}\cdot\left(\frac{\e^{\i\frac{\alpha-\beta}{2}}-\e^{-\i\frac{\alpha-\beta}{2}}}{2}M+\frac{\e^{\i\frac{\alpha-\beta}{2}}+\e^{-\i\frac{\alpha-\beta}{2}}}{2}I\right)\,.
\end{equation*}
\end{proof}

With expression~\eqref{U(M)} in hand, we can state the fact mentionned at the beginning of this section more precisely.

\begin{observation}\label{subfamilies}
The family of couplings whose $U$ have at most two eigenvalues contains the following subfamilies:
\begin{itemize}
\item \PS couplings: they can be obtained from~\eqref{U(M)} exactly for $M=-I^{(n)}+\frac{2}{n}J^{(n)}$ and $M=I^{(n)}-\frac{2}{n}J^{(n)}$;
\item scale invariant couplings: they can be obtained from~\eqref{U(M)} with $\{\alpha,\beta\}\subset\{0,\pi\}$.
\end{itemize}
\end{observation}


\section{Scattering properties}\label{Scat. prop.}

One of the main properties of the studied family that makes it physically important is a very simple way in which the scattering matrices depend on the momentum $k$.

\begin{proposition}
Let the vertex coupling be given by boundary conditions~\eqref{U-form}, where $U$ satisfies~\eqref{U(M)}, i.e.,
$$
U=\e^{\i\frac{\alpha+\beta}{2}}\cdot\left[\cos\frac{\alpha-\beta}{2}\cdot I+\i\sin\frac{\alpha-\beta}{2}\cdot M\right]\,,\quad\alpha,\beta\in(-\pi,\pi]
$$
and $M$ is a Hermitian unitary matrix.
Then
\begin{equation}\label{S-matrix2}
\S(k)=\frac{\left(k\cdot\cos\frac{\alpha}{2}\cos\frac{\beta}{2}+\frac{1}{k}\cdot\sin\frac{\alpha}{2}\sin\frac{\beta}{2}\right)\cdot I+\i\sin\frac{\alpha-\beta}{2}\cdot M}{k\cdot\cos\frac{\alpha}{2}\cos\frac{\beta}{2}-\frac{1}{k}\cdot\sin\frac{\alpha}{2}\sin\frac{\beta}{2}-\i\sin\frac{\alpha+\beta}{2}}\,.
\end{equation}
\end{proposition}

\begin{proof}
Formula~\eqref{S-matrix2} can be obtained from Equation~\eqref{S-matrix} easily with the aid of the spectral decomposition~\eqref{U(P)} of $U$.
\end{proof}

Therefore, if $U$ has at most two eigenvalues, the corresponding scattering matrix~\eqref{S-matrix2} has the form
\begin{equation}\label{munu}
\S(k)=\mu(k)I+\nu(k)M\,,
\end{equation}
where $M$ is a Hermitian unitary matrix. It means that
\begin{itemize}
\item the off-diagonal part of $\S(k)$ is uniformly scaled with $k$,
\item the diagonal part of $\S(k)$ is a sum of two uniformly scaled components, $I$ and $\diag(M)$.
\end{itemize}

\begin{corollary}\label{TransRat}
Let a vertex coupling be given by boundary conditions~\eqref{U-form}, where $U$ satisfies~\eqref{U(M)}.
Then for any $j,\ell$ ($j\neq\ell$) and $j',\ell'$ ($j'\neq\ell'$) such that $M_{j'\ell'}\neq0$, the quantity
$$
\frac{|\S(k)_{j\ell}|}{|\S(k)_{j'\ell'}|}
$$
is constant with respect to $k$. Consequently, the ratios of transmission probabilities are independent of energy.
\end{corollary}


\section{Equally-transmitting couplings}\label{GPS couplings}

The main goal of this section, and also of the whole paper, is to describe a wide family of equally-transmitting vertex couplings, for which results of Section~\ref{Scat. prop.} will be used.

If we start from~\eqref{munu}, we immediately obtain:

\begin{observation}\label{constdiag}
Let a vertex coupling be given by boundary conditions~\eqref{U-form}, where
$$
U=\e^{\i\frac{\alpha+\beta}{2}}\cdot\left[\cos\frac{\alpha-\beta}{2}\cdot I+\i\sin\frac{\alpha-\beta}{2}\cdot M\right]\,,\quad\alpha,\beta\in(-\pi,\pi]
$$
and $M$ is a Hermitian unitary matrix. Then the coupling is equally-transmitting if and only if $M$ is \MPS.
\end{observation}

Let us note that if $M$ is a diagonal matrix, then $U$ is diagonal as well, which corresponds to a decoupled vertex (cf. Rem.~\ref{decoupled}). For this reason we assume from now on that
\begin{equation}\label{assumption}
\text{$M$ is a \emph{non-diagonal} Hermitian unitary \MPS matrix}\,.
\end{equation}

In the next step we express the ratio $\rho(k)=\frac{|\S(k)_{jj}|^2}{|\S(k)_{j\ell}|^2}$, defined in~\eqref{rho}, in terms of $k$ and $d=\frac{|M_{jj}|}{|M_{j\ell}|}$.

\begin{proposition}\label{rhoProp}
Let a vertex couplings be given by boundary conditions~\eqref{U-form}, where $U$ is given by~\eqref{U(M)} and $M$ satisfies~\eqref{assumption}, let $d=\frac{|M_{jj}|}{|M_{j\ell}|}$. Then
\begin{equation}\label{rho(k)}
\rho(k)=d^2+(d^2+n-1)\cdot\frac{\left(k\cdot\cos\frac{\alpha}{2}\cos\frac{\beta}{2}+\frac{1}{k}\cdot\sin\frac{\alpha}{2}\sin\frac{\beta}{2}\right)^2}{\sin^2\frac{\alpha-\beta}{2}}\,.
\end{equation}
\end{proposition}

\begin{proof}
From~\eqref{S-matrix2} we obtain
\begin{multline*}
\rho(k)=\frac{|k\cdot\cos\frac{\alpha}{2}\cos\frac{\beta}{2}+\frac{1}{k}\cdot\sin\frac{\alpha}{2}\sin\frac{\beta}{2}+\i\sin\frac{\alpha-\beta}{2}\cdot M_{jj}|^2}{|\i\sin\frac{\alpha-\beta}{2}\cdot M_{j\ell}|^2}
=\frac{1}{|M_{j\ell}|^2}\cdot\frac{\left(k\cdot\cos\frac{\alpha}{2}\cos\frac{\beta}{2}+\frac{1}{k}\cdot\sin\frac{\alpha}{2}\sin\frac{\beta}{2}\right)^2}{\sin^2\frac{\alpha-\beta}{2}}+d^2\,,
\end{multline*}
and since $M$ is unitary and \MPS, it holds $\frac{1}{|M_{j\ell}|^2}=\frac{|M_{jj}|^2+(n-1)|M_{j\ell}|^2}{|M_{j\ell}|^2}=d^2+n-1$.
\end{proof}

With regard to Proposition~\ref{rhoProp}, possible behaviours of $\rho(k)$ can be classified into four classes according to $\alpha$ and $\beta$; let us discuss them one by one.

\subsection*{Scattering of type I: Scale invariant coupling}

If $\{\alpha,\beta\}=\{0,\pi\}$, i.e., $\sigma(U)=\{1,-1\}$, we have
$$
\S(k)=\pm M \quad \text{for all $k\in[0,+\infty)$}\,,
$$
which represents a scale invariant coupling. Its reflection and transmission probabilities are independent of $k$ and
$$
\rho(k)=d^2=\textrm{const.}
$$

\subsection*{Scattering of type II: Generalized $\delta$-coupling}

If $\alpha\notin\{0,\pi\}$ and $\beta=\pi$ (alternatively vice versa), i.e., $\sigma(U)=\{\e^{\i\alpha},-1\}$, we obtain from~\eqref{rho(k)}
$$
\rho(k)=d^2+(d^2+n-1)\,\tg^2\frac{\alpha}{2}\cdot\frac{1}{k^2}\,.
$$
Note that if $M$ with certain $d$ is fixed and $|\alpha|$ ranges the interval $(0,\pi)$, then the factor $c_d(\alpha):=(d^2+n-1)\tg^2\frac{\alpha}{2}$ standing in front of $\frac{1}{k^2}$ ranges $(0,+\infty)$. Consequently, with a given $M$, we can construct a vertex coupling whose scattering matrix is \MPS and the ratio $\rho(k)$ satisfies
$$
\rho(k)=d^2+c\cdot\frac{1}{k^2} \qquad\text{for any desired $c>0$}\,.
$$
System of this type is fully reflecting, $\S(k)\sim -I$, at $k\to0$. Physically speaking, this can be identified as \emph{generalized $\delta$-coupling}, see also~\cite{CET09}. Prominent 
representatives 
of this family are:
\begin{itemize}
\item $\delta$-coupling, obtained for $M=-I+\frac{2}{n}J$,
\item $\delta_p$-coupling, obtained for $M=I-\frac{2}{n}J$.
\end{itemize}
It can be shown that the parameter $\gamma$ of the $\delta$-coupling (cf. Table~\ref{DeltaCouplings}), as well as the parameter $\gamma$ of the $\delta_p$-coupling, is related to $\alpha$ in the way $\gamma=-n\cdot\tg\frac{\alpha}{2}$.

Let us define \emph{principal parameter} $\gamma:=-n\cdot\tg\frac{\alpha}{2}$ for any generalized $\delta$-coupling determined by an $n\times n$ matrix $U$ satisfying $\sigma(U)=\{\e^{\i\alpha},-1\}$. Apparently, the principal parameter of the standard $\delta$- and $\delta_p$-coupling coincides with their parameters from Table~\ref{DeltaCouplings}. All the generalized $\delta$-couplings with the same value of the principal parameter have many common physical properties, for example, it can be easily checked that the corresponding star graphs with infinite edges have the same point spectrum.

\subsection*{Scattering of type III: Generalized $\delta'$-coupling}

If $\alpha=0$ and $\beta\notin\{0,\pi\}$ (alternatively vice versa), i.e., $\sigma(U)=\{1,\e^{\i\beta}\}$, we obtain
$$
\rho(k)=d^2+(d^2+n-1)\,\cotg^2\frac{\beta}{2}\cdot k^2\,.
$$
In the same way as in the previous case, we deduce that with a given $M$, one can construct a vertex coupling whose scattering matrix is \MPS and
$$
\rho(k)=d^2+c\cdot k^2 \qquad\text{for any desired $c>0$}\,.
$$
System of this type is fully reflecting, $\S(k)\sim I$, at $k\to\infty$. Physically speaking, this can be identified as generalized $\delta'$-coupling. Its prominent 
representatives 
are (cf. Table~\ref{DeltaCouplings}):
\begin{itemize}
\item $\delta'$-coupling, obtained for $M=-I+\frac{2}{n}J$,
\item $\delta'_s$-coupling, obtained for $M=I-\frac{2}{n}J$.
\end{itemize}
It can be shown that the parameter $\gamma'$ of the $\delta'$-coupling (cf. Table~\ref{DeltaCouplings}), as well as the parameter $\gamma'$ of the $\delta'_s$-coupling, is related to $\beta$ in the way $\gamma'=-n\cdot\cotg\frac{\beta}{2}$.

We define \emph{principal parameter} $\gamma':=-n\cdot\cotg\frac{\beta}{2}$ for any generalized $\delta'$-coupling determined by an $n\times n$ matrix $U$ satisfying $\sigma(U)=\{1,\e^{\i\beta}\}$. The principal parameter of the standard $\delta'$- and $\delta'_s$-coupling coincides with their parameters from Table~\ref{DeltaCouplings}.

\subsection*{Scattering of type IV: Mixed coupling}

If none of the numbers $\alpha,\beta$ belongs to $\{0,\pi\}$, i.e., if $\sigma(U)\cap\{1,-1\}=\emptyset$, we introduce the parameter $\xi\in(-\frac{\pi}{2},\frac{\pi}{2})\backslash\{0\}$ by the relation
\begin{equation}\label{xi}
\tg\xi=\frac{\sin\frac{\alpha}{2}\sin\frac{\beta}{2}}{\cos\frac{\alpha}{2}\cos\frac{\beta}{2}}=\tg\frac{\alpha}{2}\tg\frac{\beta}{2}\,;
\end{equation}
then
\begin{equation}\label{rhoIV}
\rho(k)=d^2+(d^2+n-1)\cdot\frac{\cos^2\frac{\alpha}{2}\cos^2\frac{\beta}{2}+\sin^2\frac{\alpha}{2}\sin^2\frac{\beta}{2}}{\sin^2\frac{\alpha-\beta}{2}}\left(\cos\xi\cdot k+\sin\xi\cdot\frac{1}{k}\right)^2\,.
\end{equation}
This system is fully reflecting, $\S(k)\sim \pm I$, both at $k\to0$ and $k\to\infty$. Considering the previously discussed types of scattering, we can say that the coupling with the characteristics~\eqref{rhoIV} is in a way a mixture of the $\delta$ and $\delta'$ type couplings, where the $\delta/\delta'$ ratio is characterized by $|\tg\xi|$. Note that since $\alpha$ and $\beta$ run over $(-\pi,\pi)\backslash\{0\}$, $|\tg\xi|$ can take any positive value, thus the components of type $\delta$ and $\delta'$ can be mixed in any proportion.

Similarly as for scatterings of type II and III, let us discuss the range of the factor
$$
c:=(d^2+n-1)\cdot\frac{\cos^2\frac{\alpha}{2}\cos^2\frac{\beta}{2}+\sin^2\frac{\alpha}{2}\sin^2\frac{\beta}{2}}{\sin^2\frac{\alpha-\beta}{2}}\,,
$$
if $d$ and $\xi$ are fixed and $\alpha,\beta$ satisfy~\eqref{xi}. We need this lemma:

\begin{lemma}\label{range c}
Let $\xi\in(-\frac{\pi}{2},\frac{\pi}{2})\backslash\{0\}$ be given by~\eqref{xi}. Then the set
$$
\mathcal{F}_\xi:=\left\{\frac{\cos^2\frac{\alpha}{2}\cos^2\frac{\beta}{2}+\sin^2\frac{\alpha}{2}\sin^2\frac{\beta}{2}}{\sin^2\frac{\alpha-\beta}{2}}\;\left|\;\alpha,\beta\in(-\pi,\pi)\backslash\{0\}\,,\,\alpha\neq\beta\,,\,\tg\frac{\alpha}{2}\tg\frac{\beta}{2}=\tg\xi\right.\right\}
$$
equals $(0,+\infty)$ for $\tg\xi>0$ and $\left(0,\frac{1+\tg^2\xi}{4|\tg\xi|}\right]$ for $\tg\xi<0$.
\end{lemma}

\begin{proof}
At first, we use~\eqref{xi} to eliminate $\beta$:
$$
\frac{\cos^2\frac{\alpha}{2}\cos^2\frac{\beta}{2}+\sin^2\frac{\alpha}{2}\sin^2\frac{\beta}{2}}{\sin^2\frac{\alpha-\beta}{2}}=\frac{1+\tg^2\xi}{\left(\tg\frac{\alpha}{2}-\frac{\tg\xi}{\tg\frac{\alpha}{2}}\right)^2}\,;
$$
now $\alpha\in(-\pi,\pi)\backslash\{0\}$ is the only free parameter.
Then a trivial analysis of this expression as a function of $\alpha$ leads to the sought result:
\begin{itemize}
\item If $\tg\xi>0$, then
$$
\left\{\left.\frac{1+\tg^2\xi}{\left(\tg\frac{\alpha}{2}-\frac{\tg\xi}{\tg\frac{\alpha}{2}}\right)^2}\;\right|\;|\alpha|\in(0,\pi)\right\}=\left(0,+\infty\right)\,.
$$
\item If $\tg\xi<0$, then
$$
\left\{\left.\frac{1+\tg^2\xi}{\left(\tg\frac{\alpha}{2}-\frac{\tg\xi}{\tg\frac{\alpha}{2}}\right)^2}\;\right|\;|\alpha|\in(0,\pi)\right\}=\left(0,\frac{1+\tg^2\xi}{4|\tg\xi|}\right]\,.
$$
\end{itemize}
\end{proof}

The result of Lemma~\ref{range c} can be applied in the following way. Let us have an \MPS Hermitian unitary matrix with $d=\frac{|M_{jj}|}{|M_{j\ell}|}$. Then for any $\xi\in(-\frac{\pi}{2},\frac{\pi}{2})\backslash\{0\}$ we can use $M$ to construct an equally-transmitting vertex coupling with
$$
\rho(k)=d^2+c_{d,\xi}\left(\sin\xi\cdot k+\cos\xi\cdot\frac{1}{k}\right)^2\,,
$$
where $c_{d,\xi}$ can be chosen as any number satisfying
\begin{itemize}
\item $c_{d,\xi}\in(0,+\infty)$ for $\tg\xi>0$,
\item $c_{d,\xi}\in\left(0,(d^2+n-1)\frac{1+\tg^2\xi}{4|\tg\xi|}\right]$ for $\tg\xi<0$.
\end{itemize}

\subsubsection*{On the matrix $M$ and its parameter $d$}

In all types of scattering, $\rho(k)$ depends on the parameter $d$ of the matrix $M$.
\MPS Hermitian unitary matrices, which are essential objects in our considerations, have been studied in~\cite{SL82} and~\cite{TC11}. Let us recall the most important facts concerning their existence.

\emph{1. The range of $d$.} For any fixed order of $M$, the values of the parameter $d$ are bounded, cf.~\cite{TC11}:

\begin{proposition}\label{range of r}
Let $M$ be a Hermitian unitary \MPS matrix of order $n$, $d=\frac{|M_{jj}|}{|M_{j\ell}|}$. If $n>2$, then $d\leq\frac{n}{2}-1$.
\end{proposition}

\emph{2. Existence of an $M$ for a fixed $d$.} It is known that an \MPS Hermitian unitary $n\times n$ matrix $M$ with the maximal value $d_{\max}=\frac{n}{2}-1$ exists for any $n$, one can take for example $M=-I^{(n)}+\frac{2}{n}J^{(n)}$ or $M=I^{(n)}-\frac{2}{n}J^{(n)}$. As for $d<\frac{n}{2}-1$, the matrices $M$ have been so far characterized only for certain values of $d$, for example:
\begin{itemize}
\item There exists an $M$ with $d\in\left[\frac{n}{2}-3,\frac{n}{2}-1\right)$ for all even $n\in\N$.
\item There exists an $M$ with $d=\frac{n}{4}-\frac{3}{2}$ for all even $n\in\N$.
\end{itemize}
The structure of these matrices, as well as some results for smaller $d$, can be found in~\cite{TC11}. Let us refer also to the paper~\cite{SL82}, where real matrices $M$ are investigated in detail.

Let us mention that the question of characterization of \MPS unitary matrices is not only important for our problem, but is interesting in itself, as it is related to several open problems in matrix theory and combinatorics, such as to the Hadamard conjecture.


\section{Conclusions}

Vertex couplings corresponding to unitary matrices having at most two eigenvalues are characteristic by a simple $k$-dependence of their scattering matrices $\S(k)$. That property allows, using e.g. the method described in~\cite{CET10}, 
after transforming the connection condition into ``$ST$-form'',
to construct quantum vertices with $k$-independent ratios of transmission probabilities.  In this paper, we have focused on equally-transmitting couplings, i.e., couplings for which a probability of the transmission from the $j$-th edge to $\ell$-th edge is independent of $(j,\ell)$ for all momenta $k$, and furthermore, all the reflection probabilities are equal, regardless of the value of $k$. We have discovered two their special subfamilies that can be understood, with regard to their physical properties, as generalizations of the classical $\delta$-coupling and $\delta'$-coupling.

The construction of a vertex coupling with prescribed scattering properties relies on the existence of a certain Hermitian unitary matrix $M$. This fact may serve as a physical motivation for a study of Hermitian unitary matrices with specific relations between the moduli of their elements.

Our results also inspire an interesting question, whether
or not
there exist equally-transmitting couplings different from those we have described here.

\section*{Acknowledgments}
This research was supported  by the Japan Ministry of Education, Culture, Sports, Science and Technology under the Grant number 21540402.



\end{document}